\def\be{\begin{equation}}
	\def\ee{\end{equation}}
\def\ba{\begin{array}}
	\def\ea{\end{array}}

\def\mathbi#1{\text{\em #1}}

\documentclass[prl,showpacs,twocolumn,amsmath]{revtex4}
\usepackage{amsmath}
\usepackage{amsfonts}
\usepackage{mathrsfs}
\usepackage{amssymb}
\usepackage{pifont}
\usepackage{epsfig,subfigure,dsfont,amsthm,amsbsy,mathrsfs,amscd}
\usepackage{epstopdf}
\usepackage{bbm}
\usepackage{color}
\def\qed{\leavevmode\unskip\penalty9999 \hbox{}\nobreak\hfill
	\quad\hbox{\leavevmode  \hbox to.77778em{%
			\hfil\vrule   \vbox to.675em%
			{\hrule width.6em\vfil\hrule}\vrule\hfil}}
	\par\vskip3pt}
\usepackage{leftidx}
\newtheorem{theorem}{Theorem}
\newtheorem{corollary}{Corollary}
\newtheorem{lemma}{Lemma}

\input amssym.def
\begin{document}
	\title{\large\bf Improving quantum battery capacity in tripartite quantum systems by local projective measurements}
	
	\author{Yiding Wang$^{1,*}$, Hui Liu$^{1,*}$, Shao-Ming Fei$^{2}$ and Tinggui Zhang$^{1,\dag}$ }
	\affiliation{$1$ School of Mathematics and Statistics, Hainan Normal University, Haikou, 571158, China \\$2$ School of Mathematical Sciences, Capital Normal University, Beijing 100048, China\\
		$^{*}$ First co-authors (These authors contributed equally to this work) \\
		$^{\dag}$ Correspondence to tinggui333@163.com}
	
	\bigskip
	\bigskip
	
	\begin{abstract}
We investigate the impact of local von Neumann measurements on quantum battery capacity in tripartite quantum systems. We propose two measurement-based protocols and introduce the concept of optimal local projective operators. Specifically, we derive explicit analytical expressions for our protocols when applied to general three-qubit $X$-states. Furthermore, we analyze the negative effects of white noise and dephasing noise on quantum battery capacity, proving that optimal local projective operators can improve the robustness of subsystem and total system capacity against both noise types for the general tripartite $X$-state. We numerically validate the performance of different schemes in capacity enhancement through detailed examples and find that these optimized operators can effectively enhance both subsystem and total system battery capacity. Our results indicate that the local von Neumann measurement is a powerful tool to enhance the battery capacity in multipartite quantum systems.
	\end{abstract}
	
	\pacs{04.70.Dy, 03.65.Ud, 04.62.+v} \maketitle
	
	\section{I. Introduction}
	As quantum devices can store and release energy in an appropriate manner \cite{phhs,cmov,akmc,mfha,omk}, quantum batteries have enormous potential to outperform the classical batteries. Due to the small size, high portability and large capacity, they can be employed in nanotechnology such as nano-chips. In recent years, quantum batteries have been widely studied in the field of quantum science.
	
	The concept of quantum batteries in an informational theoretic context was first introduced by R. Alicki and M. Fannes \cite{rm}, and characterized by the maximum amount of energy that can be extracted from a quantum system by applying unitary operations. Afterwards, quantum batteries have been extensively studied both theoretically and experimentally \cite{ffh,bfh,cpbc,sjo,fgpp,afmp,llmp,fcap,nfmh,rap,akmg,fb,ktss,gcs,gccb,jqqw,rarc,drdr,ys,xzzl,mtmj,ksus,zts,gcms,zdz,sdww,dlws,shsg,fqdy,sgas,fmaj,cmha,dzs,tscg}. The authors in \cite{cpbc} proved that quantum mechanics can contribute to an enhancement in the amount of work deposited per unit time. A classification of several models of quantum batteries was given by \cite{afmp}, involving different combinations of two-level systems and quantum harmonic oscillators. In \cite{llmp}, the authors extended the concept of quantum battery from a collection of a priori isolated systems to a many-body quantum system with intrinsic interactions. The problems of many-body interaction and energy fluctuation were also studied in \cite{rap,ys,xzzl,zts,dlws,fqdy}. Ferraro et al. used collective charging to prove that quantum batteries charge faster than ordinary batteries and proposed the first quantum battery model that could be engineered in a soild-state architecture \cite{fcap}. Barra demonstrated that a cyclic unitary process can extract work from the thermodynamic equilibrium state of an engineered quantum dissipative process \cite{fb}. In addition, some protocols have been presented to maximize charging efficiency and prevent energy losses due to the effect of environment \cite{nfmh,drdr,gccb,mtmj,jqqw,shsg}. Moreover, an important quantitative indicator of the quality for quantum batteries, i.e., the ergotropy functional, was provided by Tirone et al. \cite{tscg}.
	
	Recently, new progresses have been made in the research of quantum batteries \cite{xzzz,yyas,strs,ksm,yswz,bjbf,cbss,ssyw,cadm,zglc,masf,zmjy,zyf,risj}. Xu et al. found that increasing the two-mode coupling can enhance the charging power \cite{xzzz}. In \cite{yyas} the quantum battery capacity is defined to be the difference between the highest and the lowest energy that can be reached under the unitary evolution of the system. Furthermore, the author in \cite{strs} showed that employing entangled input states with local recovery operations will generally not improve the battery performance. Chaki et al. demonstrated that stochastically extractable energy using non-positive operator-valued measurements is no less than that using positive operator-valued measurements \cite{cbss}. In \cite{cadm}, the authors proved that when the driving is strong enough the useful work which can be extracted from the quantum battery, i.e., the ergotropy, is exactly equal to the stored energy.
	
	Quantum measurements on one subsystem may have an effect on the states of other subsystems. The idea to achieve the desired effect on the target system by performing appropriate quantum measurements on other subsystems has been applied to the research of quantum battery \cite{yj,zyf}. In \cite{zyf}, the authors showed that for bipartite systems the battery capacity with respect to one subsystem can be improved by local-projective measurements on another subsystem. In this study we explore the improvement of quantum battery capacity in tripartite systems. It is basically different from the bipartite scenario for which we have only one scheme: measure one subsystem and observe the capacity changes of the entire system and the another subsystem. However, in a tripartite system, we may either measure only one subsystem and observe the battery capacity changes of the other two, or measure two subsystems and observe the capacity of the remaining one. Exploring the effect of local-projective measurements on battery capacity in tripartite systems is an interesting problem. Consequently, the study on optimal local-projective operator that maximizes the increase of battery capacity is of significance.
	
	The rest of this paper is organized as follows. In section II, we provide the main results of this paper, including two measurement schemes which only involve local-projective measurements in tripartite systems [\mathbi{Scheme 1-2}\,], the concept of the optimal local projective operators [Theorem 1-2\,], the negative impact of noise on battery capacity, and the weakening of this negative impact by our schemes [Theorem 3-4\,]. In section III, we present two examples to explore the changes in battery capacity of the post-measurement subsystem and the whole battery system [\mathbi{Example 1-2}\,]. We summarize and discuss our conclusions in the last section.
	
	\section{II. Measurement-based protocol in tripartite system}
	
    Recently, a new quantification of quantum battery capacity was introduced \cite{yyas}, which admits a simple expression in terms of the eigenvalues $\{\lambda_i\}$ of the quantum state $\rho$ and the energy levels $\{\epsilon_i\}$ of the Hamiltonian $H=\sum_{i=0}^{d-1}\epsilon_i|\varepsilon_i\rangle\langle\varepsilon_i|$,
	\begin{equation}\label{e1}
		\begin{split}
			\mathcal{C}(\rho;H)&=\sum_{i=0}^{d-1}\epsilon_i(\lambda_i-\lambda_{d-1-i})\\
			&=\sum_{i=0}^{d-1}\lambda_i(\epsilon_i-\epsilon_{d-1-i}),
		\end{split}
	\end{equation}
	where $\{\lambda_i\}$ and $\{\epsilon_i\}$ are arranged in ascending order without losing generality, $\lambda_0\leqslant\lambda_1\leqslant...\leqslant\lambda_{d-1}$ and $\epsilon_0\leqslant\epsilon_1\leqslant...\leqslant\epsilon_{d-1}$.
	The battery capacity $\mathcal{C}(\rho;H)$ is unitarily invariant and is a Schur-covex functional of the quantum state. By definition, unitary invariance is clear, and Schur convexity implies that if a state $\rho$ is majorized by $\varrho$ $(\rho\prec\varrho)$, then we have $\mathcal{C}(\rho;H)\leq\mathcal{C}(\varrho;H)$.
	
	The authors in \cite{zyf} show that this quantum battery capacity with respect to a subsystem or the whole bipartite system can be improved by local-projective measurements on another subsystem. In this work, we focus on the tripartite scenario of improving battery capacity by local projective measurements.
	\\
	
	Let us consider the battery state $\rho_{ABC}$ in Hilbert space $\mathcal{H}_A\otimes\mathcal{H}_B\otimes\mathcal{H}_C$, with Hamiltonians $H_A=\sum_{i=0}^{d_A-1}\epsilon_i^A|\varepsilon_i\rangle^A\langle\varepsilon_i|$, $H_B=\sum_{i=0}^{d_B-1}\epsilon_i^B|\varepsilon_i\rangle^B\langle\varepsilon_i|$ and $H_C=\sum_{i=0}^{d_C-1}\epsilon_i^C|\varepsilon_i\rangle^C\langle\varepsilon_i|$ associated with systems $A$, $B$ and $C$, respectively. The Hamiltonian of the entire system is
	\begin{small}
	\begin{equation}\label{e2}
	\begin{split}
	H&=H_A\otimes I_{d_B}\otimes I_{d_C}+I_{d_A}\otimes H_B\otimes I_{d_C}+I_{d_A}\otimes I_{d_B}\otimes H_C\\
	&+J_{AB}\otimes I_{d_C}+J_{AC}\otimes I_{d_B}+J_{BC}\otimes I_{d_A}+J_{ABC},
	\end{split}
	\end{equation}
	\end{small}
	 where $I_{d_X}$ denotes the $d_X\times d_X$ identity matrix in system $X$, $J_{XY}$ represents the interaction between subsystem $X$ and $Y$, and $J_{ABC}$ represents the interaction between subsystem $A$, $B$ and $C$. Now we propose the first measurement scheme.
	 \\
	
	 \mathbi{Scheme 1}. In this strategy, we research how local-projective measurements on subsystem $C$ would affect the amount of capacity of the subsystem $AB$ or the whole tripartite battery system.
	
	 From (\ref{e1}), the quantum battery capacity with respect to system $AB$ is given by
	 \begin{equation}\label{e3}
	 	\mathcal{C}(\rho_{AB};H_{AB})=\sum_{i=0}^{d_Ad_B-1}\epsilon_i^{AB}
	 	(\lambda_i^{AB}-\lambda_{d_Ad_B-1-i}^{AB}),
	 \end{equation}
	 where $H_{AB}=H_A\otimes I_{d_B}+I_{d_A}\otimes H_B+J_{AB}$ is the Hamiltonian of the subsystem $AB$, $\{\epsilon_i^{AB}\}$ and $\{\lambda_i^{AB}\}$ are the energy levels of $H_{AB}$ and the eigenvalues of the reduced state $\rho_{AB}$, respectively. Let $\{W_k\}_{k=0}^{d_C-1}$ be a rank-1 von Neumann measurement on the subsystem $C$. Conditioned on the measurement outcome, say $k$, the post-measurement state of $\rho_{ABC}$ becomes
	 \begin{equation}\label{e4}
	 	\rho_{ABC}^k=\frac{1}{P_k}(I_{d_A}\otimes I_{d_B}\otimes W_k)\rho_{ABC}(I_{d_A}\otimes I_{d_B}\otimes W_k)^\dagger,
	 \end{equation}
	 where the probability
	 $$
	 P_k=\text{Tr}[(I_{d_A}\otimes I_{d_B}\otimes W_k)\rho_{ABC}(I_{d_A}\otimes I_{d_B}\otimes W_k)^\dagger]
	 $$
	 with $k=0,1,...,d_C-1$. In fact, the final state depends on the measurement basis and the outcomes. Next, we can calculate the capacity of $\rho_{AB}^k$ according to the measurement outcome $k$,
	 \begin{equation}\label{e5}
	 \mathcal{C}(\rho_{AB}^k;H_{AB})=\sum_{i=0}^{d_Ad_B-1}\epsilon_i^{AB}
	 (\xi_i^{AB}-\xi_{d_Ad_B-1-i}^{AB}),
	 \end{equation}
	 where $\{\xi_i^{AB}\}$ are the eigenvalues of the reduced state $\rho_{AB}^k$.
	
	 Now, we consider how to choose a suitable local-projective operator that maximizes the battery capacity of subsystem $AB$ in the final state.  A local-projection operator is called optimal local-projective operator if it improves the battery capacity of the subsystem more than any other projection operator. From now we restrict all measurement bases in this work to the computational basis. Based on measurement scheme 1, our results are the following.
	 \begin{theorem}
	 	If the state $\rho_{AB}^i\,(i=1,\dots,d_C)$ is majorized by $\rho_{AB}^k$, then the local measurement $I_{d_A}\otimes I_{d_B}\otimes W_k$ corresponding to the state $\rho_{ABC}^k$ is the optimal local-projective operator, where $k\in \{1,2,\cdots,d_C\}$.
	 \end{theorem}
	
	 \begin{proof}
	 	From Ref.\cite{yyas}, the battery capacity is a Schur-convex functional of the state. So according to $\rho_{AB}^i\,(i=1,\dots,d_C)$ is majorized by $\rho_{AB}^k$, we can obtain that
	 	$$
	 	\mathcal{C}(\rho_{AB}^i;H_{AB})\leq\mathcal{C}(\rho_{AB}^k;H_{AB}),\,i=1,\dots,d_C.
	 	$$
	 \end{proof}
	 In summary, the target subsystem's battery capacity obtained via optimal local projective operator. We use the optimal projection operator to implement scheme 1 as shown in Figure 1.
	 \begin{figure}[htbp]
	 	\centering
	 	\includegraphics[width=0.5\textwidth]{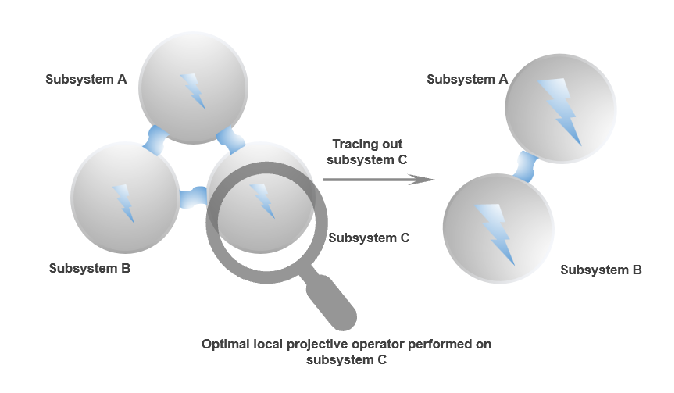}
	 	\vspace{-2em} \caption{The idea of tripartite measurement-based protocol for \mathbi{Scheme 1} is shown in the Figure. We measure the $C$ system locally and observe its influence on the battery capacity of the whole or the remaining systems $AB$.} \label{Fig.1}
	 \end{figure}
	
	 In scheme 1, the local-projective measurement on one subsystem is used. Thus it is a natural idea to consider the local-projective measurement on two subsystems.
	 \\
	
	 \mathbi{Scheme 2}. Here we consider the variations of the battery capacities of the subsystem $A$ and the whole tripartite system under the von Neumann measurements on the subsystem $BC$.
	
	 For the initial state $\rho_{ABC}$, the battery capacity with respect to subsystem $A$ is given by
	 \begin{equation}\label{e6} \mathcal{C}(\rho_{A};H_{A})=\sum_{i=0}^{d_A-1}\epsilon_i^{A}(\lambda_i^{A}-\lambda_{d_A-1-i}^{A}).
	 \end{equation}
	 After the von Neumann measurement $\{V_k\}_{k=0}^{d_Bd_C-1}$ on the subsystem $BC$, with respect to the $k$th measurement outcome the state $\rho_{ABC}$ becomes
	 \begin{equation}\label{e7}
	 	\rho_{ABC}^k=\frac{1}{P_k}(I_{d_A}\otimes V_k)\rho_{ABC}(I_{d_A}\otimes V_k)^\dagger,
	 \end{equation}
	 where the probability
	 $$
	 P_k=\text{Tr}[(I_{d_A}\otimes V_k)\rho_{ABC}(I_{d_A}\otimes V_k)^\dagger]
	 $$
	 with $k=0,1,...,d_Bd_C-1$. Then, we can calculate the capacity of $\rho_{A}^k$ corresponding to the measurement outcome $k$,
	 \begin{equation}\label{e8}
	 	\mathcal{C}(\rho_{A}^k;H_{A})=\sum_{i=0}^{d_A-1}\epsilon_i^{A}
	 	(\xi_i^{A}-\xi_{d_A-1-i}^{A}),
	 \end{equation}
	 where $\{\xi_i^{A}\}$ are the eigenvalues of the reduced state $\rho_{A}^k$. Similar to scheme 1, we can also find a suitable local-projective operator that maximizes the battery capacity of subsystem $A$ in the final state.
	 \begin{theorem}
	 	The operator $I_{d_A}\otimes V_k$ corresponding to the state $\rho_{ABC}^k$ is the optimal local-projective operator to improve the battery capacity of subsystem $A$, if the state $\rho_{A}^i\,(i=1,\dots,d_Bd_C)$ is majorized by $\rho_{A}^k$.
	 \end{theorem}
	
      The implementation scheme 2 using the optimal projection operator is shown in Figure 2.
	 \begin{figure}[htbp]
	 	\centering
	 	\includegraphics[width=0.45\textwidth]{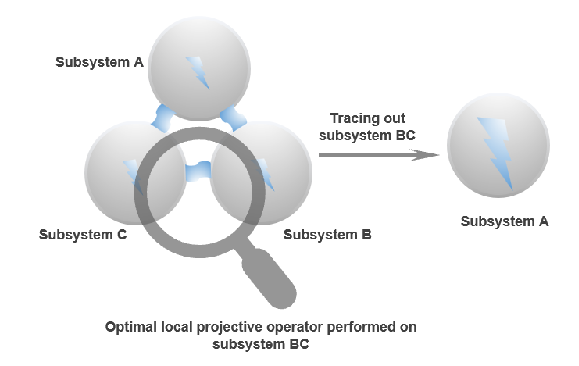}
	 	\vspace{-1em} \caption{The measurement-based protocol in \mathbi{Scheme 2}: we measure the $BC$ system locally to observe its influence on the battery capacity of the whole and the subsystem $A$.} \label{Fig.2}
	 \end{figure}
	
	 Our measurement-based protocol has specific analytical expressions for the general three qubits $X$ state in the computational basis, see Appendix A for details. In practical environments, some key operations and measurements on quantum systems rely on classical instrumentation, which inevitably introduce white noise. In fact, it can be proved that white noise also has a negative impact on the battery capacity of general state.
	 \begin{theorem}
	 	White noise can reduce the capacity of the $n$-dimensional battery state $\rho$, that is, $\mathcal{C}(\rho;H)\geq\mathcal{C}((1-f)\rho+f\mathbbm{1}/n;H)\,(f\in[0,1])$, and $\mathcal{C}((1-f)\rho+f\mathbbm{1}/n;H)$ is the monotonic decreasing functional of the noise intensity $f$.
	 \end{theorem}
	 See Appendix B for the proof details of Theorem 3. From Theorem 3, we can conclude that white noise can weaken the battery capacity of any quantum state. Given a battery state $\rho$ with Hamiltonian $H$, we can regard $\mathcal{C}(\rho;H)$ as the rated battery capacity of $\rho$, and define the capacity loss rate $r_l$ as the ratio of the battery capacity reduced by white noise to the rated battery capacity, i.e.
	 \begin{small}
	 	\begin{equation}\label{e11}
	 		r_l=\frac{\mathcal{C}(\rho;H)-\mathcal{C}((1-f)\rho+f\mathbbm{1}/n;H)}{\mathcal{C}(\rho;H)},
	 	\end{equation}
	 \end{small}
	 where $f\in[0,1]$ is the noise intensity. Furthermore, the dephasing channel represents another fundamental noise model in quantum information theory \cite{bypb}, primarily characterizing the loss of phase information in qubits during transmission or storage. Given the fundamental role of $X$-states (including Bell states, Werner states, and Greenberger-Horne-Zeilinger (GHZ) states as special cases) in quantum information, a critical and interesting question is whether our protocol can enhance their battery capacity robustness against both white noise and dephasing channels. We can obtain the following fact, see Appendix C for proof.
	 \begin{theorem}
	 For any tripartite $X$-state, if the optimal local-projective operator can be identified within our protocol, the target subsystem's and whole system's capacity robustness against the white noise can be improved. For dephasing noise, our optimal local projective operators enable complete robustness for the target subsystem's and the whole system's capacity.
	 \end{theorem}
	
	\section{III. Applications of measurement-based protocol}
	In this section, we use specific examples to verify the effectiveness of our two measurement-based protocols.
	
	 \mathbi{Example 1.} Consider the W state \cite{xbzk,vngaa,tttw,xywyc} mixed with white noise,
	 $$
	 \rho_a=(1-a)|W\rangle\langle W|+\frac{a}{8}\mathbbm{1},
	 $$
	 where $|W\rangle=1/\sqrt{3}(|001\rangle+|010\rangle+|100\rangle)$, $a\in[0,1]$, and $\mathbbm{1}$ denotes the $8\times8$ identity matrix. We take the Hamiltonian of the whole system as
	 \begin{small}
	 	\begin{equation}\label{e9}
	 		\begin{split}
	 			H&=\epsilon^A\sigma_3\otimes I_2\otimes I_2+\epsilon^BI_2\otimes\sigma_3\otimes I_2+\epsilon^CI_2\otimes I_2\otimes \sigma_3\\
	 			&+J_{AB}\sigma_1\otimes\sigma_1\otimes I_2+J_{AC}\sigma_1\otimes I_2\otimes\sigma_1+J_{BC}I_2\otimes\sigma_1\otimes\sigma_1,\\
	 		\end{split}
	 	\end{equation}
	 \end{small}
	 and the Hamiltonian with respect to subsystem $AB$ is
	 \begin{small}
	 \begin{equation}\label{e10}
	 H_{AB}=\epsilon^A\sigma_3\otimes I_2+\epsilon^BI_2\otimes\sigma_3+J_{AB}\sigma_1\otimes\sigma_1,
	 \end{equation}
	 \end{small}
	 where $\epsilon^A\geq\epsilon^B\geq\epsilon^C>0$, $\sigma_i$ $(i=1,2,3)$ are the standard Pauli operators, $I_2$ is the $2\times2$ identity matrix, and $J_{XY}\geq0\,(XY=AB,BC,AC)$ are the interaction parameters between subsystems. We use $\{\epsilon_i\}_{i=0}^{7}$ and $\{\epsilon_i^{AB}\}_{i=0}^3$ to represent the eigenvalues of $H$ and $H_{AB}$, respectively, which depend on $\epsilon^A$, $\epsilon^B$, $\epsilon^C$ and the interaction parameters between subsystems. It can be calculate that the eigenvalues of $\rho_a$ are $\lambda_i=a/8\,(i=0,\dots,6),\lambda_7=1-7a/8$, and the eigenvalues of $\rho_a^{AB}$ are $\lambda_i^{AB}=a/4\,(i=0,1), \lambda_2^{AB}=(4-a)/12,\,\lambda_3^{AB}=(8-5a)/12$. So the initial battery capacity of the whole system and subsystems are
	 \begin{small}
	 \begin{equation*}
	 \begin{split}
	 &\mathcal{C}(\rho_a;H)=(1-a)\epsilon_7+(a-1)\epsilon_0,\\
	 &\mathcal{C}(\rho_a^{AB};H_{AB})=\frac{2-2a}{3}(\epsilon_3^{AB}-\epsilon_0^{AB})+\frac{1-a}{3}(\epsilon_2^{AB}-\epsilon_1^{AB}),\\
	 &\mathcal{C}(\rho_a^A;H_A)=\frac{2-2a}{3}\epsilon^A.
	 \end{split}
	 \end{equation*}
	 \end{small}
	
	 Now we implement the above two measurement schemes to investigate the impact of local projective measurement on the battery capacity of the subsystem or the entire system. For the convenience of comparison, in the following numerical calculations we set $\epsilon^A=0.5$, $\epsilon^B=0.3$ and $\epsilon^C=0.1$, and default $J_{XY}=0.1\,(XY=AB,BC,AC)$ without special declaration.
	
	 For scheme 1, we consider the projective measurement $\{W_k=|k\rangle\langle k|\}_{k=0}^1$ on subsystem $C$. From Eq.\,(\ref{e4}), one can obtain $\rho_a^{0,AB}$ and $\rho_a^{1,AB}$. Note that $\rho_a^{0,AB}\succ\rho_a^{1,AB}$, so the optimal local projective operator is $I_2\otimes I_2\otimes W_0$ in scheme 1 according to Theorem 1. When we use the optimal projection operator to implement scheme 1, it can be calculated that  $\mathcal{C}(\rho_a^{0,AB};H_{AB})\geq\mathcal{C}(\rho_a^{AB};H_{AB})$ and $\mathcal{C}(\rho_a^0;H)\geq\mathcal{C}(\rho_a;H)$, which means that the battery capacities of the subsystem $AB$ and the whole system have been improved. When the noise intensity reaches the maximum, the quantum state degenerates to the maximally mixed state, and the whole and the subsystems' battery capacity are all $0$. It is worth noting that our scheme can not improve the battery capacity of the subsystem when $a=1$, but it can improve the battery capacity of the entire system, see Figure 3.
	 \begin{figure}[htbp]
	 	\centering
	 	\includegraphics[width=0.5\textwidth]{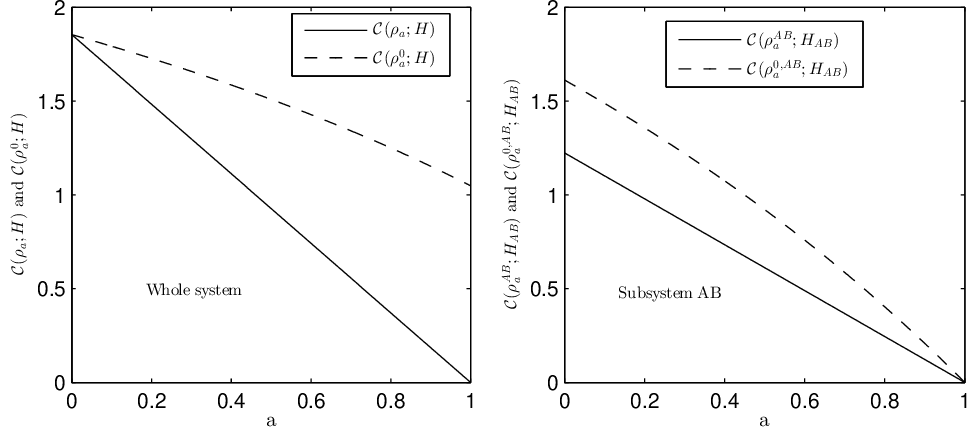}
	 	\vspace{-2em} \caption{Under fixed parameters $J_{AB}=J_{BC}=J_{AC}=0.1$, the total system's battery capacity and the subsystem AB's capacity as the functions of the noise strength $a\,(0\sim1)$, before and after the measurement.} \label{Fig.3}
	 \end{figure}
	
	 For scheme 2, we use the projective measure $\{V_k\}_{k=0}^3=\{|00\rangle\langle00|,|01\rangle\langle01|,|10\rangle\langle10|,|11\rangle\langle11|\}$ on subsystem $BC$. Based on Eq.(\ref{e7}), we can obtain $\rho_a^{i,A}\,(i=0,1,2,3)$. In this case, $I_2\otimes V_i\,(i=0,1,2)$ are all optimal local projective operators. Therefore, we have
	 \begin{small}
	 \begin{equation*}
	 \mathcal{C}(\rho_a^{0,A};H_A)/\mathcal{C}(\rho_a^{A};H_A)=12/(4-a)>1.
	 \end{equation*}
	 \end{small}
	 In other words, the optimal local projective operator greatly enhances the battery capacity of subsystem $A$.
	
	 Now we study the influence of the optimal local projective operator under two schemes on the total system battery capacity. We use $r_i$ to represent the ratio of the battery capacity of the final state to that of the initial state in scheme $i$ ($i=1,2$) and study the impact of changes in interaction parameters $J_{AB}$, $J_{AC}$, $J_{BC}$, and noisy intensity $a$ on $r_1$ and $r_2$ (for example, when the value of $J_{AB}$ changes, $J_{BC}$ and $J_{AC}$ are fixed at $0.1$), see Figure 4.
	 \begin{figure*}[htpb]
	 	\centering
	 	\includegraphics[width=\textwidth,height=0.53\textwidth]{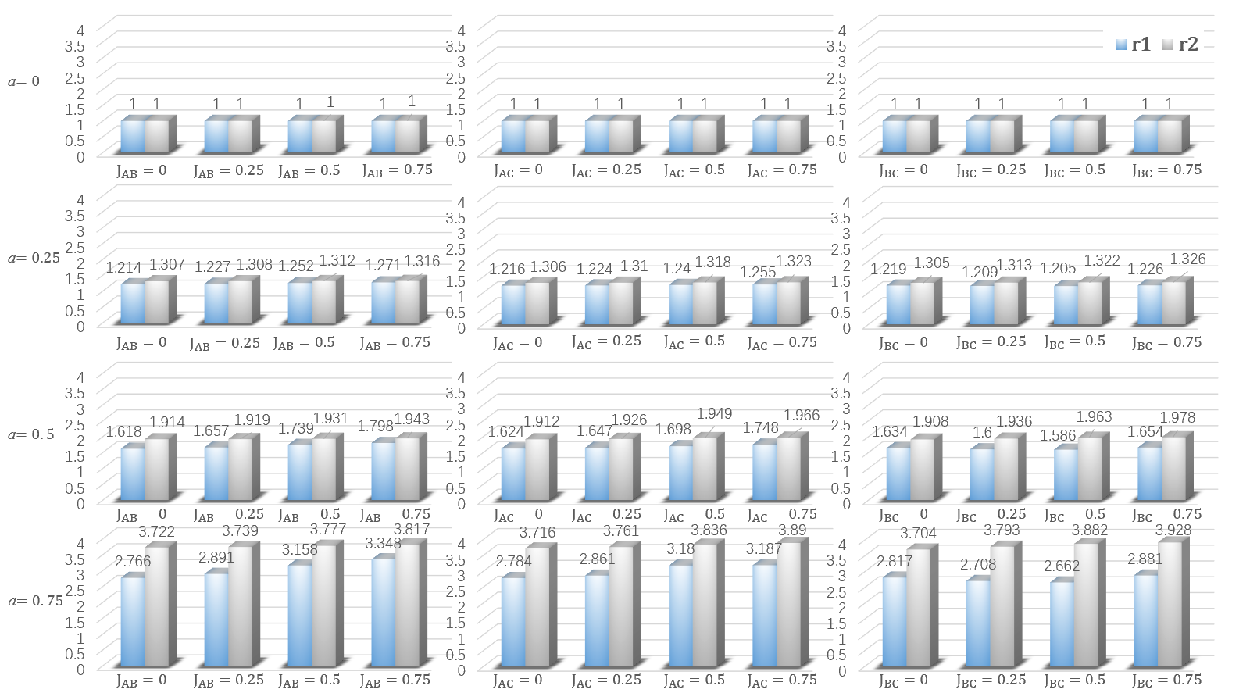}
	 	\vspace{-2em} \caption{Figure 4 contains 12 sub-plots arranged in a $3\times4$ grid. Sub-plots in the same row share identical noise strength ($a=0, 0.25, 0.5, 0.75$), while those in the same column share the same interaction strength configuration ($J_{XY}=0, 0.25, 0.5, 0.75$). For instance, the sub-plot in the second row and first column demonstrates how the varying $J_{AB}$ values ($0, 0.25, 0.5, 0.75$) affect the battery capacity enhancement ratios $r_1$ and $r_2$ under fixed parameters $a=0.25$ and $J_{BC}=J_{AC}=0.1$.} \label{Fig.4}
	 \end{figure*}
	 We find that no matter which interaction parameter increases, $r_2$ is always not lower than $r_1$. This means that the local-projective measurements on two subsystems are more effective than those on one subsystem in improving the entire battery capacity. The reason behind this fact is that the state obtained by using the optimal operator in scheme 1 has at most 4 non-zero eigenvalues, while the state corresponding to scheme 2 has at most 2 non-zero eigenvalues. As the trace of a density matrix is 1, mathematically the former is more easily majorized by the latter, and thus the latter has a larger battery capacity. Furthermore, when the noise intensity $a=0$, our two measurement schemes increase the battery capacity of subsystem $AB$ and subsystem $A$ respectively without reducing the whole system battery capacity. With the increase of noise intensity $a$, the efficiency $r_1$ and $r_2$ of the two schemes are improved. This phenomenon is reasonable, because the white noise weakens the capacity by increasing the number of non-zero eigenvalues of the battery state. However, both of our measurement schemes can reduce the number of non-zero eigenvalues to enhance the entire system battery capacity. In other words, our measurement schemes can weaken the negative impact of noise on the battery capacity, so they have anti-noise performance.	
\\	

 As the last example in this paper, let us consider the GHZ state used in quantum tasks \cite{kchk,lsbl,gjmg}.
	 \\
	
	 \mathbi{Example 2.} Let the initial state is three-qubit GHZ state $|GHZ\rangle=1/\sqrt{2}\sum_{i=0}^{1}|i\rangle^{\otimes 3}$ affected by the white noise,
	 \begin{equation*}
	 	\rho_b=(1-b)|GHZ\rangle\langle GHZ|+\frac{b}{8}\mathbbm{1},
	 \end{equation*}
	 where $0\leq b\leq1$. The whole system Hamiltonian we considered is
	 \begin{small}
	 	\begin{equation}\label{e12}
	 		\begin{split}
	 			H&=\epsilon^A\sigma_3\otimes I_2\otimes I_2+\epsilon^BI_2\otimes\sigma_3\otimes I_2+\epsilon^CI_2\otimes I_2\otimes \sigma_3\\
	 			&+J_{ABC}\sigma_1\otimes\sigma_1\otimes\sigma_1,
	 		\end{split}
	 	\end{equation}
	 \end{small}
	 and the Hamiltonian corresponding to subsystem $AB$ is
	 \begin{small}
	 	\begin{equation}\label{e13}
	 		H_{AB}=\epsilon^A\sigma_3\otimes I_2+\epsilon^BI_2\otimes\sigma_3,
	 	\end{equation}
	 \end{small}
	 where $\epsilon^A\geq\epsilon^B\geq\epsilon^C>0$ and $J_{ABC}\geq0$ is the interaction parameter between subsystem $A$, $B$ and $C$. The reason why Eqs.\,(\ref{e12}) and (\ref{e13}) are different from previous Eqs.\,(\ref{e9}) and (\ref{e10}) is that the quantum coherence of the three qubits $X$ state is shared by subsystem $A$, $B$ and $C$, which can be seen from the fact that $\rho_x^{AB}$, $\rho_x^{AC}$ and $\rho_x^{BC}$ are incoherent states.
	
	 From the analytical expressions in Appendix A, we can obtain that the eigenvalues of $\rho_b$ are $\lambda_i=b/8\,(i=0,\dots,6)$, $\lambda_7=1-7b/8$, and the eigenvalues of $\rho_b^{AB}$ are $\lambda_i^{AB}=b/4\,(i=0,1)$, $\lambda_j=(2-b)/4\,(j=2,3)$. As mention above, We also use $\{\epsilon_i\}_{i=0}^{7}$ to represent the eigenvalues of $H$, which depend on $\epsilon^A$, $\epsilon^B$, $\epsilon^C$ and the interaction parameter $J_{ABC}$. Thus the initial battery capacity of the entire system and subsystems are
	 \begin{small}
	 	\begin{equation*}
	 		\begin{split}
	 			&\mathcal{C}(\rho_b;H)=(1-b)\epsilon_7+(b-1)\epsilon_0,\\
	 			&\mathcal{C}(\rho_a^{AB};H_{AB})=2(1-b)\epsilon^A,\\
	 			&\mathcal{C}(\rho_a^A;H_A)=0.
	 		\end{split}
	 	\end{equation*}
	 \end{small}
	
	 We first implement scheme 1 to improve the battery capacity of subsystem $AB$. According to Corollary 1 in Appendix A, one can find that both $I_2\otimes I_2\otimes W_0$ and $I_2\otimes I_2\otimes W_1$ are optimal local projective operators in this case. Then using operator $I_2\otimes I_2\otimes W_0$ to implement scheme 1, we have $\mathcal{C}(\rho_b^{0,AB};H_{AB})=2(1-b)(\epsilon^A+\epsilon^B)>\mathcal{C}(\rho_b^{AB};H_{AB})$. This means that the battery capacity of the subsystem $AB$ have been enhanced.
	
	 For the scheme 2, it can be verified that the optimal local projective operator is $I_2\otimes V_0$ or $I_2\otimes V_3$ based on Corollary 2 in Appendix A. So we have $\mathcal{C}(\rho_b^{0,A};H_A)=\mathcal{C}(\rho_b^{3,A};H_A)=\frac{4-4b}{2-b}\epsilon^A$. That is to say, the optimal local projective operators in scheme 2 improve the battery capacity of subsystem $A$ in this situation.
	
	 Now let us compare scheme 1 and scheme 2 to investigate the influence of optimal local-projective operator on total system battery capacity under three measurement schemes. In this case, we can define the ratio of the improved value of the optimal local-projective operator in the $i$-th measurement scheme for the battery capacity to the rated battery capacity $\mathcal{C}(\rho_0;H$) as the capacity recovery rate $r_i^+$ and research the impact of changes in interaction parameter $J_{ABC}$ and noise parameter $b$ on capacity loss rate $r_l$ and capacity recovery rate $r_i^+\,(i=1,2)$, see Figure 5.
	 \begin{figure}[htbp]
	 	\centering
	 	\includegraphics[width=0.5\textwidth]{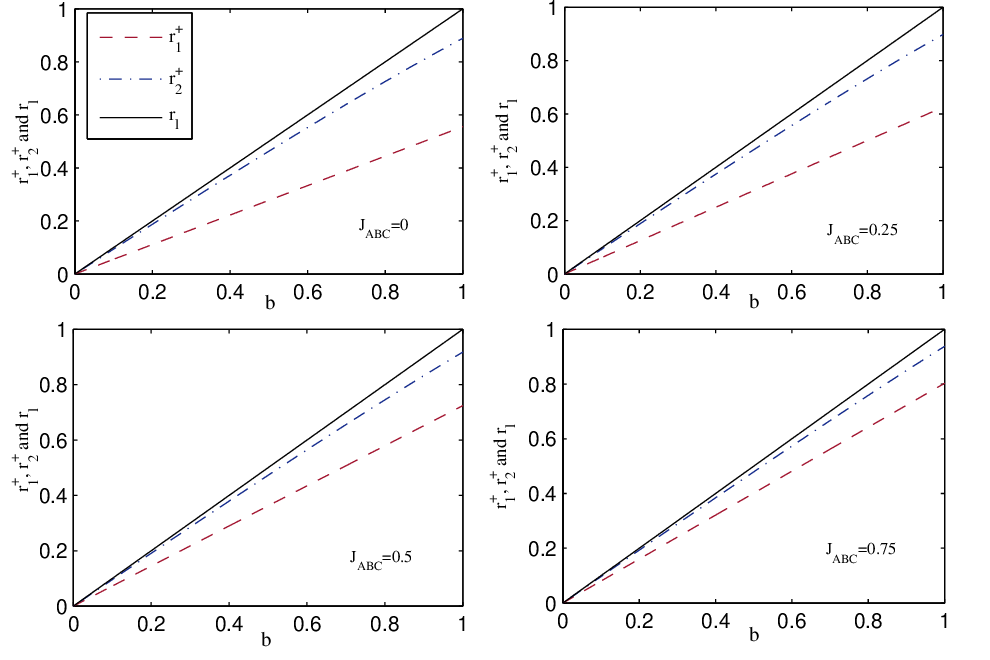}
	 	\vspace{-2em} \caption{The impact of changes in parameters $J_{ABC}\,(0,0.25,0.5,0.75)$ and $b\,(0\sim1)$ on $C_1$ and $C_2$ in Example 2.} \label{Fig.5}
	 \end{figure}
	
We observe that the capacity loss rate $r_l$ increase with the enhance of noise intensity $b$. In particular, when the noise intensity reaches the maximum $b=1$, the capacity loss rate $r_l=1$. These results are not difficult to understand mathematically: the monotonic increasing property of $r_l$ can be obtained by combining Eq.\,(\ref{e11}) and the monotonic decreasing property of battery capacity with respect to noise intensity from Theorem 3. When $b=1$, it can be verified that $\rho_1=\mathbbm{1}/8$ has a minimum battery capacity of $0$, and thus $r_l$ reaches the maximum. Furthermore, it can be found that the capacity recovery rate $r_2^+$ is always not lower than $r_1^+$. This means that although both of our measurement schemes enhance the battery capacity, scheme 2 shows better capacity recovery ability than scheme 1. The reasons for this phenomenon are as follows: when the noise parameter $b>0$, $\rho_b$ has $8$ non-zero eigenvalues. However, the local-projective measurement in scheme 1 and scheme 2 can reduce the number of its non-zero eigenvalues to $4$ and $2$, respectively. So $\rho_b$ is majorized by the post-measurement state in scheme 1, and the post-measurement state in scheme 1 is majorized by the counterpart in scheme 2. Therefore, the battery capacity of the post-measurement state in scheme 1 is greater than the $\rho_b$'s battery capacity, and the battery capacity of the post-measurement state in scheme 2 is greater than the counterpart in scheme 1 according to the fact that battery capacity is a Schur convex functional. In summary, the interference of white noise will weaken the battery capacity, but our measurement schemes can recovery part of the battery capacity. In addition, the above example also tells us that our schemes can enhance the battery capacity of the subsystem or the entire tripartite system, even if the initial state is not an entangled one.
	
	\section{IV. Conclusions and discussions}
	We have investigated the effect of local-projective measurements on the quantum battery capacity in tripartite systems by adopting two schemes based on the measured subsystems and introduced the concept of optimal local-projective operators for the tripartite states. Additionally, we have presented explicit analytical expressions for our schemes when applied to general three-qubit $X$-states. Moreover, we have studied the adverse effects of white noise and dephasing noise on battery capacity and proved that our measurement schemes can improve the robustness of target subsystem and total system battery capacity against both noise types for the general tripartite $X$-state. Finally, we have numerically validated the performance of different schemes in capacity enhancement through concrete examples and found that the optimal local-projective operators could be more effective in enhancing the battery capacity of both subsystem and total system.
	
	Our results show that the local von Neumann measurement is a useful tool to improve the battery capacity in multipartite quantum systems. In fact, both positive and non-positive measurements have been taken into account in energy extraction from quantum batteries \cite{cbss}. It would be interesting to further investigate how our scheme affects the battery capacity under the unitary time evolution of the entire system. The effects of more general positive operator-valued measurement\,(POVM) or nonlocal joint measurements on the enhancement of battery capacity in multipartite quantum systems are also worthwhile to be investigated.
	
	\bigskip
	{\bf Acknowledgments:} The authors thank the anonymous reviewers for their specific comments, which really improved the quality of this paper. This work is supported by the National Natural Science Foundation of China (NSFC) under Grant Nos. 12204137 and 12171044; the specific research fund of the Innovation Platform for Academicians of Hainan Province.

\appendix
\section{Appendix A: Three-qubit $X$ states}
Consider the initial state be the three-qubit $X$ states in the computational basis,
$$
\rho_x=\left(\begin{array}{cccccccc}
	\rho_{11} & 0 & 0 & 0 & 0 & 0 & 0 & \rho_{18} \\
	0 & \rho_{22} & 0 & 0 & 0 & 0 & \rho_{27} & 0 \\
	0 & 0 & \rho_{33} & 0 & 0 & \rho_{36} & 0 & 0 \\
	0 & 0 & 0 & \rho_{44} & \rho_{45} & 0 & 0 & 0 \\
	0 & 0 & 0 & \rho_{54} & \rho_{55} & 0 & 0 & 0 \\
	0 & 0 & \rho_{63} & 0 & 0 & \rho_{66} & 0 & 0 \\
	0 & \rho_{72} & 0 & 0 & 0 & 0 & \rho_{77} & 0 \\
	\rho_{81} & 0 & 0 & 0 & 0 & 0 & 0 & \rho_{88}
\end{array}\right),
$$
where $\sum_{i=1}^{8}\rho_{ii}=1$ and $\rho_{ii}\geq0$. The eigenvalues of $\rho_x$ are
\begin{equation*}
	\begin{split}
		 &\lambda_0=\frac{1}{2}(\rho_{11}+\rho_{88}-\sqrt{(\rho_{11}-\rho_{88})^2+4|\rho_{18}|^2}),\\
		 &\lambda_1=\frac{1}{2}(\rho_{11}+\rho_{88}+\sqrt{(\rho_{11}-\rho_{88})^2+4|\rho_{18}|^2}),\\
		 &\lambda_2=\frac{1}{2}(\rho_{22}+\rho_{77}-\sqrt{(\rho_{22}-\rho_{77})^2+4|\rho_{27}|^2}),\\
		 &\lambda_3=\frac{1}{2}(\rho_{22}+\rho_{77}+\sqrt{(\rho_{22}-\rho_{77})^2+4|\rho_{27}|^2}),\\
		 &\lambda_4=\frac{1}{2}(\rho_{33}+\rho_{66}-\sqrt{(\rho_{33}-\rho_{66})^2+4|\rho_{36}|^2}),\\
		 &\lambda_5=\frac{1}{2}(\rho_{33}+\rho_{66}+\sqrt{(\rho_{33}-\rho_{66})^2+4|\rho_{36}|^2}),\\
		 &\lambda_6=\frac{1}{2}(\rho_{44}+\rho_{55}-\sqrt{(\rho_{44}-\rho_{55})^2+4|\rho_{45}|^2}),\\
		 &\lambda_7=\frac{1}{2}(\rho_{44}+\rho_{55}+\sqrt{(\rho_{44}-\rho_{55})^2+4|\rho_{45}|^2}).
	\end{split}
\end{equation*}
Consider the following Hamiltonian of the whole system,
\begin{small}
	\begin{equation*}
	\begin{split}
	H_{ABC}&=\epsilon^A\sigma_3\otimes I_2\otimes I_2+\epsilon^BI_2\otimes\sigma_3\otimes I_2+\epsilon^CI_2\otimes I_2\otimes\sigma_3\\
	&+J_{ABC}\sigma_1\otimes\sigma_1\otimes\sigma_1,
	\end{split}
	\end{equation*}
\end{small}
where $\epsilon^A\geqslant\epsilon^B\geqslant\epsilon^C\geqslant0$ and $J_{ABC}\geq0$ is the interaction parameter between subsystems $A$, $B$ and $C$.

For the scheme 1, we perform the von Neumann measurement given by the computational bases $\{W_k=|k\rangle\langle k|, k=0,1\}$ on subsystem $C$. For simplicity, we can write $\rho_x$ in the Bloch representation,
\begin{equation*}
	\begin{split}
		\rho_x=&\frac{1}{8}(I_2\otimes I_2\otimes I_2+a_1\sigma_3\otimes I_2\otimes I_2+b_1I_2\otimes\sigma_3\otimes I_2\\
		&+c_1I_2\otimes I_2\otimes\sigma_3+a_2I_2\otimes\sigma_3\otimes\sigma_3+b_2\sigma_3\otimes I_2\otimes\sigma_3\\
		&+c_2\sigma_3\otimes\sigma_3\otimes I_2+a_3\sigma_1\otimes\sigma_2\otimes\sigma_2+b_3\sigma_2\otimes\sigma_1\otimes\sigma_2\\
		&+c_3\sigma_2\otimes\sigma_2\otimes\sigma_1
		+e_1\sigma_1\otimes\sigma_1\otimes\sigma_1+e_3\sigma_3\otimes\sigma_3\otimes\sigma_3),
	\end{split}
\end{equation*}
where
\begin{equation*}
	\begin{split}
		 &a_1=\rho_{11}+\rho_{22}+\rho_{33}+\rho_{44}-\rho_{55}-\rho_{66}-\rho_{77}-\rho_{88},\\
		 &b_1=\rho_{11}+\rho_{22}+\rho_{55}+\rho_{66}-\rho_{33}-\rho_{44}-\rho_{77}-\rho_{88},\\
		 &c_1=\rho_{11}+\rho_{33}+\rho_{55}+\rho_{77}-\rho_{22}-\rho_{44}-\rho_{66}-\rho_{88},\\
		 &a_2=\rho_{11}+\rho_{44}+\rho_{55}+\rho_{88}-\rho_{22}-\rho_{33}-\rho_{66}-\rho_{77},\\
		 &b_2=\rho_{11}+\rho_{33}+\rho_{66}+\rho_{88}-\rho_{22}-\rho_{44}-\rho_{55}-\rho_{77},\\
		 &c_2=\rho_{11}+\rho_{22}+\rho_{77}+\rho_{88}-\rho_{33}-\rho_{44}-\rho_{55}-\rho_{66},\\
		&a_3=2(\rho_{27}+\rho_{36}-\rho_{18}-\rho_{45}),\\
		&b_3=2(\rho_{27}+\rho_{45}-\rho_{18}-\rho_{36}),\\
		&c_3=2(\rho_{36}+\rho_{45}-\rho_{18}-\rho_{27}),\\
		&e_1=2(\rho_{27}+\rho_{36}+\rho_{18}+\rho_{45}),\\
		&e_3=\rho_{11}+\rho_{44}+\rho_{66}+\rho_{77}-\rho_{22}-\rho_{33}-\rho_{55}-\rho_{88}.
	\end{split}
\end{equation*}
Performing the projection measurement on system $C$, we obtain the measurement outcome ensemble $\{\rho_x^k,P_k\}_{k=0}^1$, where
\begin{equation*}
	\begin{split}
		P_k\rho_x^k=&(I_2\otimes I_2\otimes W_k)\rho_x(I_2\otimes I_2\otimes W_k)^\dagger\\
		=&\frac{1}{8}(I_2\otimes I_2\otimes W_k+a_1\sigma_3\otimes I_2\otimes W_k+b_1I_2\otimes\sigma_3\otimes W_k\\
		&+(-1)^kc_1I_2\otimes I_2\otimes W_k+(-1)^ka_2I_2\otimes\sigma_3\otimes W_k\\
		&+(-1)^kb_2\sigma_3\otimes I_2\otimes W_k+c_2\sigma_3\otimes\sigma_3\otimes W_k\\
		&+(-1)^ke_3\sigma_3\otimes\sigma_3\otimes W_k)
	\end{split}
\end{equation*}
with $P_0=\frac{1+c_1}{2}$ and $P_1=\frac{1-c_1}{2}$.
From Eqs.(\ref{e1}) and (\ref{e5}), one may calculate the battery capacity of the post-measurement state of the subsystem $AB$ and the entire tripartite system.

From Theorem 1, we can also consider how to choose a suitable local-projective operator for the general three-qubit $X$ state. In fact, $\rho_x^0$ and $\rho_x^1$ are both diagonal, with at most four non-zero elements (eigenvalues) each.
\begin{corollary}
	The local measurement $I_2\otimes I_2\otimes W_0$ corresponding to the state $\rho_x^0$ is the optimal local-projective measurement for $\rho_x$ to improve the battery capacity of subsystem $AB$, if the state $\rho_x^{1,AB}$ is majorized by $\rho_x^{0,AB}$, and vice versa.
\end{corollary}

Now we consider  the three-qubit $X$-state in scheme 2. After the measurements
$
\{V_k\}_{k=0}^3=\{|00\rangle\langle00|,|01\rangle\langle01|,|10\rangle\langle10|,|11\rangle\langle11|\}
$
on system $BC$, one has
\begin{small}
	\begin{equation*}
		\begin{split}
			P_0\rho_x^0&=(I_2\otimes V_0)\rho_x(I_2\otimes V_0)^\dagger\\
			&=(I_2\otimes W_0\otimes W_0)\rho_x(I_2\otimes W_0\otimes W_0)\\
			&=\frac{1}{8}(I_2\otimes W_0\otimes W_0+a_1\sigma_3\otimes W_0\otimes W_0+b_1I_2\otimes W_0\otimes W_0\\
			&+c_1I_2\otimes W_0\otimes W_0+a_2I_2\otimes W_0\otimes W_0+b_2\sigma_3\otimes W_0\otimes W_0\\
			&+c_2\sigma_3\otimes W_0\otimes W_0+e_3\sigma_3\otimes W_0\otimes W_0),
		\end{split}
	\end{equation*}
	\begin{equation*}
	\begin{split}
		P_1\rho_x^1&=(I_2\otimes V_1)\rho_x(I_2\otimes V_1)^\dagger\\
		&=(I_2\otimes W_0\otimes W_1)\rho_x(I_2\otimes W_0\otimes W_1)\\
		&=\frac{1}{8}(I_2\otimes W_0\otimes W_1+a_1\sigma_3\otimes W_0\otimes W_1+b_1I_2\otimes W_0\otimes W_1\\
		&-c_1I_2\otimes W_0\otimes W_1-a_2I_2\otimes W_0\otimes W_1-b_2\sigma_3\otimes W_0\otimes W_1\\
		&+c_2\sigma_3\otimes W_0\otimes W_1-e_3\sigma_3\otimes W_0\otimes W_1),
	\end{split}
\end{equation*}
\begin{equation*}
	\begin{split}
		P_2\rho_x^2&=(I_2\otimes V_2)\rho_x(I_2\otimes V_2)^\dagger\\
		&=(I_2\otimes W_1\otimes W_0)\rho_x(I_2\otimes W_1\otimes W_0)\\
		&=\frac{1}{8}(I_2\otimes W_1\otimes W_0+a_1\sigma_3\otimes W_1\otimes W_0-b_1I_2\otimes W_1\otimes W_0\\
		&+c_1I_2\otimes W_1\otimes W_0-a_2I_2\otimes W_1\otimes W_0+b_2\sigma_3\otimes W_1\otimes W_0\\
		&-c_2\sigma_3\otimes W_1\otimes W_0-e_3\sigma_3\otimes W_1\otimes W_0),
	\end{split}
\end{equation*}
\begin{equation*}
	\begin{split}
		P_3\rho_x^3&=(I_2\otimes V_3)\rho_x(I_2\otimes V_3)^\dagger\\
		&=(I_2\otimes W_1\otimes W_1)\rho_x(I_2\otimes W_1\otimes W_1)\\
		&=\frac{1}{8}(I_2\otimes W_1\otimes W_1+a_1\sigma_3\otimes W_1\otimes W_1-b_1I_2\otimes W_1\otimes W_1\\
		&-c_1I_2\otimes W_1\otimes W_1+a_2I_2\otimes W_1\otimes W_1-b_2\sigma_3\otimes W_1\otimes W_1\\
		&-c_2\sigma_3\otimes W_1\otimes W_1+e_3\sigma_3\otimes W_1\otimes W_1),
	\end{split}
\end{equation*}
\end{small}
where $W_i=|i\rangle\langle i|$, $i=0,1$, and the probabilities
\begin{small}
	\begin{equation*}
		\begin{split}
			&P_0=\frac{1+b_1+c_1+a_2}{4},~~P_1=\frac{1+b_1-c_1-a_2}{4},\\
			&P_2=\frac{1-b_1+c_1-a_2}{4},~~P_3=\frac{1-b_1-c_1+a_2}{4}.
		\end{split}
	\end{equation*}
\end{small}

Note that each $\rho_x^i$ $(i=0,1,2,3)$ has at most two non-zero eigenvalues. From Theorem 2, we have the following corollary related to the optimal local-projective operator in scheme 2.

\begin{corollary}
	The operator $I_2\otimes V_k\,(k\in\{0,1,2,3\})$ corresponding to the state $\rho_x^k$ is the optimal local-projective operator to improve the battery capacity of subsystem $A$, if the state $\rho_x^{i,A}\,(i=0,1,2,3)$ is majorized by $\rho_x^{k,A}$.
\end{corollary}

\section{Appendix B: Proof of Theorem 3}
\setcounter{equation}{0}
\renewcommand{\theequation}{B\arabic{equation}}
	For an $n$-dimensional state $\rho$, let its eigenvalues be denoted as $\lambda_0\leq\lambda_2\leq\dots\leq\lambda_{n-1}$. So the eigenvalues of $(1-f)\rho+f\mathbbm{1}/n$ are
	\begin{small}
		\begin{equation*}
			(1-f)\lambda_0+f/n\leq(1-f)\lambda_1+f/n\leq\dots(1-f)\lambda_{n-1}+f/n.
		\end{equation*}
	\end{small}
	For any non-negative integer $k<n-1$,
	\begin{small}
		\begin{equation*}
			\begin{split}
				&\,\,\,\,\,\,\,\sum_{i=0}^{k}\lambda_{n-1-i}-\sum_{i=0}^{k}[(1-f)\lambda_{n-1-i}+f/n]\\
				&=f\sum_{i=0}^{k}\lambda_{n-1-i}-(k+1)f/n\\
				&=f(1-\sum_{i=0}^{n-k-2}\lambda_i-(k+1)/n)\\
				&\geq f(1-(n-k-1)/n-(k+1)/n)\\
				&=0.
			\end{split}
		\end{equation*}
	\end{small}
	This implies that $\rho\succ(1-f)\rho+f\mathbbm{1}/n$. Thus one have $\mathcal{C}(\rho;H)\geq\mathcal{C}((1-f)\rho+f\mathbbm{1}/n;H)$ based on the fact that the battery capacity is a Schur-convex functional of the quantum state. Now we prove monotonicity. Using $\{\epsilon_i\}_{i=0}^{n-1}$ to represent the energy levels of Hamiltonian $H$. According to Eq.\,(\ref{e1}), one have
	\begin{small}
		\begin{equation*}
			\begin{split}
				\frac{\partial\,\mathcal{C}((1-f)\rho+f\mathbbm{1}/n;H)}{\partial f}&=\frac{\partial\sum_{i=0}^{n-1}\epsilon_i(1-f)(\lambda_i-\lambda_{n-1-i})}{\partial f}\\
				&=\frac{\partial(1-f)\sum_{i=0}^{n-1}\epsilon_i(\lambda_i-\lambda_{n-1-i})}{\partial f}\\
				&=\frac{\partial(1-f)\mathcal{C}(\rho;H)}{\partial f}\\
				&=-\mathcal{C}(\rho;H)\\
				&\leq0.
			\end{split}
		\end{equation*}
	\end{small}
\section{Appendix C: Proof of Theorem 4}
\setcounter{equation}{0}
\renewcommand{\theequation}{C\arabic{equation}}
Before considering noise, we first show that the optimal local projection operator can enhance the battery capacity of any tripartite $X$-state subsystem and the total system.

Let $\rho$ be a given $X$-state in the Hilbert space $\mathcal{H}_{A_1}\otimes\mathcal{H}_{A_2}\otimes\mathcal{H}_{A_3}$\,(the dimension of $\mathcal{H}_{A_i}$ is $d_i$). Suppose we use the optimal local projective operator $M_1=(|1\rangle\langle1|)^{\otimes d_k\cdots d_3}$ on subsystem $\mathcal{H}_{A_k}\otimes\cdots\otimes\mathcal{H}_{A_3}$ without losing generality\,($1<k\leq 3$). In fact, each of the $d_k\cdots d_3$ measurement bases $M_1, M_2,\dots, M_{d_k\cdots d_3}$ corresponds to an array, which is composed of eigenvalues of the reduced density matrix of the post-measurement state relative to subsystem $\mathcal{H}_{A_1}\otimes\dots\otimes\mathcal{H}_{A_{k-1}}$ in descending order. The array corresponding to the operator $M_i$ is denoted as $N_i$, i.e.
\begin{equation}
N_i=\{\lambda_1^i,\lambda_2^i,...\lambda_{d_1\cdots d_{k-1}}^i\}.
\end{equation}
For $N_i$ and $N_j$ arrays, we say that $N_i$ is majorized by $N_j$\,($N_i\prec N_j$), if
\begin{equation}
\begin{split}
&\sum_{l=1}^{d^*}\lambda_l^i<\sum_{l=1}^{d^*}\lambda_l^j,\,1\leq d^*< d_1\cdots d_{k-1}.\\
&\sum_{l=1}^{d_1\cdots d_{k-1}}\lambda_l^i=\sum_{l=1}^{d_1\cdots d_{k-1}}\lambda_l^j.
\end{split}
\end{equation}
$M_1$ is optimal implies that
\begin{equation}
N_1\succ N_i\,(i=1,2,\dots,d_k\cdots d_3).
\end{equation}
To complete our proof, we first prove a lemma.
\begin{lemma}
For two $n^*$-ary arrays $\mathcal{N}_1$ and $\mathcal{N}_2$, if $\mathcal{N}_1\succ\mathcal{N}_2$, then
\begin{equation}
\mathcal{N}_1\succ (p\mathcal{N}_1+(1-p)\mathcal{N}_2)_+\succ(p\mathcal{N}_1+(1-p)\mathcal{N}_2)_{\Phi(n^*)},
\end{equation}
where $0\leq p\leq1$, $\Phi(n^*)$ is the rearrangement of $n^*$ elements out of order, and
\begin{equation*}
\begin{split}
&(p\mathcal{N}_1+(1-p)\mathcal{N}_2)_+=\{p\lambda_i^1+(1-p)\lambda_i^2\}_{i=1}^{n^*},\\
&(p\mathcal{N}_1+(1-p)\mathcal{N}_2)_{\Phi(n^*)}=\{p\lambda_i^1+(1-p)\lambda_{\Phi(i)}^2\}_{i=1}^{n^*}.
\end{split}
\end{equation*}
\end{lemma}
\begin{proof}
We assume that $\sum_{i=1}^{n^*}\lambda_i^1=\sum_{i=1}^{n^*}\lambda_i^2=1$ without losing generality. It is obvious that disordered addition is majorized by ordered addition, thus we only need to prove $\mathcal{N}_1\succ (p\mathcal{N}_1+(1-p)\mathcal{N}_2)_+$. Note that $\mathcal{N}_1\succ\mathcal{N}_2$, so one have
\begin{equation*}
\begin{split}
\sum_{i=1}^{k^*<n^*}p\lambda_i^1+(1-p)\lambda_i^2&=p\sum_{i=1}^{k^*<n^*}\lambda_i^1+(1-p)\sum_{i=1}^{k^*<n^*}\lambda_i^2\\
&\leq p\sum_{i=1}^{k^*<n^*}\lambda_i^1+(1-p)\sum_{i=1}^{k^*<n^*}\lambda_i^1\\
&=\sum_{i=1}^{k^*<n^*}\lambda_i^1,
\end{split}
\end{equation*}
and
$$\sum_{i=1}^{n^*}p\lambda_i^1+(1-p)\lambda_i^2=p+(1-p)=1=\sum_{i=1}^{n^*}\lambda_i^1.
$$
\end{proof}
Now continue to prove the main theorem. The reduced state of $\rho$ relative to subsystem $\mathcal{H}_{A_1}\otimes\dots\otimes\mathcal{H}_{A_{k-1}}$ is
\begin{widetext}
$$
\rho_{A_1\dots A_{k-1}}=\left(\begin{array}{cccc}
	\sum_{i=1}^{d_k\cdots d_3}\rho_{ii} & 0 & \cdots & 0 \\
	0 & \sum_{i=d_k\cdots d_3+1}^{2d_k\cdots d_3}\rho_{ii} & \cdots & 0 \\
	\vdots & \vdots & \ddots & \vdots \\
	0 & 0 & 0 & \sum_{i=(d_k\cdots d_3)(d_1\cdots d_{k-1}-1)+1}^{d_1d_2d_3}\rho_{ii}
\end{array}\right).
$$
\end{widetext}
It is not difficult to observe that for each item $\rho_{ii}$ in $\sum_{i=1}^{d_k\cdots d_3}\rho_{ii}$, $\rho_{ii}/p_i$ belongs to an array $N_i$, where
$$
p_i=\sum_{j=1}^{d_1\cdots d_{k-1}}\rho_{i+(j-1)d_k\cdots d_3,i+(j-1)d_k\cdots d_3}.
$$
 Therefore, the array $\mathcal{N}$ corresponding to the eigenvalues of $\rho_{A_1\dots A_{k-1}}$ is a disordered addition of $N_1, N_2,\dots, N_{d_k\dots d_3}$. That is to say,
\begin{equation}
\mathcal{N}=(\sum_{i=1}^{d_k\dots d_3}p_iN_i)_{\Phi(d_k\dots d_3)}.
\end{equation}
According to Lemma 1, we have $N_1\succ\mathcal{N}$. This means that
$$
\mathcal{C}(\rho^1_{A_1\dots A_{k-1}};H)\geq\mathcal{C}(\rho_{A_1\dots A_{k-1}};H),
$$
where $\rho^1=M_1\rho M_1^\dagger/\text{Tr}(M_1\rho M_1^\dagger)$.

Now we prove that the optimal local projective operator $M_1$ can enhance the whole system battery capacity. We assume that $\lambda_1^1=\rho_{11}/p_1$ without losing generality. For simplicity, we record $d_1d_2d_3=d$. So if $d$ is even, the eigenvalues of $d$-dimensional $X$-state can be expressed as $(\rho_{ii}+\rho_{d+1-i,d+1-i})/2\pm\sqrt{(\rho_{ii}-\rho_{d+1-i,d+1-i})^2+4|\rho_{i,d+1-i}|^2}/2,\,i=1,\dots d/2$. if $d$ is odd, then they are $(\rho_{ii}+\rho_{d+1-i,d+1-i})/2\pm\sqrt{(\rho_{ii}-\rho_{d+1-i,d+1-i})^2+4|\rho_{i,d+1-i}|^2}/2,\,i=1,\dots \lfloor d/2\rfloor$ and $\rho_{\lceil d/2\rceil,\lceil d/2\rceil}$. In fact, for any eigenvalue of $\rho$, we have
\begin{footnotesize}
\begin{equation*}
\begin{split}
	 &\frac{(\rho_{ii}+\rho_{d+1-i,d+1-i})\pm\sqrt{(\rho_{ii}-\rho_{d+1-i,d+1-i})^2+4|\rho_{i,d+1-i}|^2}}{2}\\
	 &\leq\frac{(\rho_{ii}+\rho_{d+1-i,d+1-i})+\sqrt{(\rho_{ii}-\rho_{d+1-i,d+1-i})^2+4|\rho_{i,d+1-i}|^2}}{2}\\
	&\leq\frac{(\rho_{ii}+\rho_{d+1-i,d+1-i})+\sqrt{(\rho_{ii}+\rho_{d+1-i,d+1-i})^2}}{2}\\
	&=\rho_{ii}+\rho_{d+1-i,d+1-i}.
\end{split}
\end{equation*}
\end{footnotesize}

Let the maximum eigenvalue of $\rho$ be $\lambda_1=(\rho_{22}+\rho_{d-1,d-1})/2+\sqrt{(\rho_{22}-\rho_{d-1,d-1})^2+4|\rho_{2,d-1}|^2}/2$ without losing generality. To prove $\lambda_1^1\geq\lambda_1$, we only need to prove $\lambda_1^1\geq\rho_{22}+\rho_{d-1,d-1}$.
\begin{small}
\begin{equation*}
	\begin{split}
		\lambda_1^1-\rho_{22}-\rho_{d-1,d-1}&=\frac{\rho_{11}}{p_1}-\rho_{22}-\rho_{d-1,d-1}\\
		&=\frac{\rho_{11}-\rho_{22}p_1-\rho_{d-1,d-1}p_1}{p_1}\\
		&=\frac{\rho_{11}\sum_{i=1}^{d}\rho_{ii}-\rho_{22}p_1-\rho_{d-1,d-1}p_1}{p_1}\\
		&=\frac{\rho_{11}\sum_{j=1}^{d_k\cdots d_3}p_j-\rho_{22}p_1-\rho_{d-1,d-1}p_1}{p_1}\\
		&\geq\frac{\rho_{11}p_2-\rho_{22}p_1+\rho_{11}p_{d_k\cdots d_3-1}-\rho_{d-1,d-1}p_1}{p_1}\\
		&\geq0.
	\end{split}
\end{equation*}
\end{small}
The last inequality is due to $\rho_{11}/p_1\geq\rho_{22}/p_2$ and $\rho_{11}/p_1\geq\rho_{d-1,d-1}/p_{d_k\cdots d_3-1}$, which come from Eq.\,(C3). For any $k^*\leq d_1\cdots d_{k-1}$, the first $k^*$ eigenvalues of $\rho$ in descending order correspond to a total of $2k^*$ non repeating diagonal elements, and at most $k^*$ of these diagonal elements belong to $M_l\rho M_l^\dagger\,(1\leq l\leq d_k\cdots d_3)$ at the same time. If such a situation exists, we record these $k^*$ elements as $\{\xi_i\}_{i=1}^{k^*}$. Then
\begin{small}
	\begin{equation*}
		\begin{split}
			 \sum_{i=1}^{k^*}\lambda_i^1-\xi_i&=\sum_{i=1}^{k^*}\lambda_i^1-\sum_{i=1}^{k^*}\xi_i\\
			&\geq\sum_{i=1}^{k^*}\lambda_i^1-\sum_{i=1}^{k^*}p_l\lambda_i^l\\
			&=\sum_{i=1}^{k^*}\lambda_i^1\sum_{j=1}^{d_k\cdots d_3}p_j-\sum_{i=1}^{k^*}p_l\lambda_i^l\\
			&\geq\sum_{i=1}^{k^*}p_l\lambda_i^1-\sum_{i=1}^{k^*}p_l\lambda_i^l\\
			&\geq0,
		\end{split}
	\end{equation*}
\end{small}
where the first inequality is due to $\xi_i/p_l\in N_l\,(i=1,\dots,k^*)$, and the last inequality follows $N_1\succ N_l$. From Eq.\,(C3), the sum of the residual $k^*$ diagonal elements must be bounded by $\sum_{i=1}^{k^*}\lambda_i^1\sum_{j=1}^{d_k\cdots d_3}(p_j-p_l)$. The above process is still used for other situations, because we only need to replace $k^*$ with any positive integer less than it. Therefore, we have proved that $\sum_{i=1}^{k^*}\lambda_i^1\geq\sum_{i=1}^{k^*}\lambda_i$, which implies that $\rho^1\succ\rho$. Based on the fact that quantum battery capacity is Schur-convex, one have $\mathcal{C}(\rho^1;H_\text{total})\geq\mathcal{C}(\rho;H_\text{total})$.

\subsection{white noise}
Now we prove that our optimal local projective operator can enhance the target subsystem and the whole system capacity robustness against white noise. In other words, for the $X$-state affected by white noise $\rho_f=(1-f)\rho+f\mathbbm{1}/d_1d_2d_3$ and noise strength $f\in[0,1]$,
\begin{equation*}
\begin{split}
\mathcal{C}((\rho^1)_{f,A_1\dots A_{k-1}};H)&\geq\mathcal{C}(\rho_{f,A_1\dots A_{k-1}};H),\\
\mathcal{C}((\rho^1)_f;H_\text{total})&\geq\mathcal{C}(\rho_f;H_\text{total}),
\end{split}
\end{equation*}
where $(\rho^1)_f=(1-f)\rho^1+f\mathbbm{1}/d_1\cdots d_3$, $\rho^1=M_1\rho M_1^\dagger/\text{Tr}(M_1\rho M_1^\dagger)$ and $M_1$ is the optimal local-projective operator.

According to the proof process in Appendix C, one have $\rho^1_{A_1\dots A_{k-1}}\succ\rho_{A_1\dots A_{k-1}}$ and $\rho^1\succ\rho$. For convenience, we denote the eigenvalue arrays corresponding to $\rho^1_{A_1\dots A_{k-1}}$ and $\rho_{A_1\dots A_{k-1}}$ as
\begin{equation}
\begin{split}
&N_1=\{\lambda_1, \lambda_2,\dots,\lambda_{d_1\dots d_{k-1}}\},\\
&\mathcal{N}=\{\mu_1, \mu_2,\dots,\mu_{d_1\dots d_{k-1}}\}.
\end{split}
\end{equation}
It can be verified that after considering white noise, their eigenvalue arrays are
\begin{small}
\begin{equation}
	\begin{split}
		&N_{f,1}=\{(1-f)\lambda_1+\frac{f}{c^*},\dots,(1-f)\lambda_{c^*}+\frac{f}{c^*}\},\\
		&\mathcal{N}_f=\{(1-f)\mu_1+\frac{f}{c^*},\dots,(1-f)\mu_{c^*}+\frac{f}{c^*}\},
	\end{split}
\end{equation}
\end{small}
where $c^*=d_1\dots d_{k-1}$. Then we have
\begin{equation*}
\begin{split}
\sum_{i=1}^{k^*<c^*}[(1-f)\mu_i+\frac{f}{c^*}]&=(1-f)\sum_{i=1}^{k^*<c^*}\mu_i+k^*\frac{f}{c^*}\\
&\leq(1-f)\sum_{i=1}^{k^*<c^*}\lambda_i+k^*\frac{f}{c^*}\\
&=\sum_{i=1}^{k^*<c^*}[(1-f)\lambda_i+\frac{f}{c^*}],
\end{split}
\end{equation*}
and
\begin{equation*}
	\begin{split}
		\sum_{i=1}^{c^*}[(1-f)\mu_i+\frac{f}{c^*}]&=(1-f)\sum_{i=1}^{c^*}\mu_i+f\\
		&=(1-f)+f\\
		&=(1-f)\sum_{i=1}^{c^*}\lambda_i+f,
	\end{split}
\end{equation*}
where we use the fact that the trace of the density matrix is $1$. Thus we have $N_{f,1}\succ\mathcal{N}_f$, which means that
$$
\mathcal{C}((\rho^1)_{f,A_1\dots A_{k-1}};H)\geq\mathcal{C}(\rho_{f,A_1\dots A_{k-1}};H).
$$
In a similar way, we can prove that
$$
\mathcal{C}((\rho^1)_f;H_\text{total})\geq\mathcal{C}(\rho_f;H_\text{total}).$$

\subsection{dephasing noise}

Perturbation-induced decoherence can be modeled by one parameter completely positive trace-preserving map
\begin{equation}
\Upsilon(\rho)=(1-\gamma)\rho+\gamma\triangle(\rho),\,\gamma\in[0,1],
\end{equation}
where $\triangle(\rho)=\sum_{i,j,k}|ijk\rangle\langle ijk|\rho|ijk\rangle\langle ijk|$ is the full dephasing channel which removes all off-diagonal elements in the computational basis. In fact, dephasing noise can reduce quantum battery capacity because the decoherent state of a quantum state is always majorized by itself \cite{wyd}, that is, $\triangle(\rho)\prec\rho$. According to Lemma 1, we have $\Upsilon(\rho)\prec\rho$.

As demonstrated in Appendix B, our optimal local projection operator can enhance battery capacity for both the target subsystem and the whole system, leading to the relation
\begin{equation*}
\begin{split}
\mathcal{C}(\Upsilon(\rho^1);H_\text{total})&=\mathcal{C}((1-\gamma)\rho^1+\gamma\triangle(\rho^1);H_\text{total})\\
&=\mathcal{C}((1-\gamma)\rho^1+\gamma\rho^1;H_\text{total})\\
&=\mathcal{C}(\rho^1;H_\text{total})\\
&\geq\mathcal{C}(\rho;H_\text{total})\\
&\geq\mathcal{C}(\Upsilon(\rho);H_\text{total}),
\end{split}
\end{equation*}
where the second equation follows the fact that the post-measurement state of an $X$-state becomes diagonal based on our schemes. Additionally, the third equation above implies that the battery capacity of post-measurement state remains unaffected by dephasing noise, exhibiting complete robustness against such decoherence processes. The conclusion for the target subsystem follows naturally from $\Upsilon(\rho^1)=\rho^1$.
\end{document}